\documentclass{article}
\usepackage[utf8]{inputenc}

\usepackage{graphicx,graphics}

\usepackage{rotating}
\usepackage[twoside]{geometry}
\usepackage{epsfig}
\usepackage{indentfirst}
\usepackage[usenames,dvipsnames,svgnames,table]{xcolor}
\usepackage{cmap}

\usepackage{graphicx,graphics}

\usepackage{hyperref}
\hypersetup{
pdftitle={Hydrogen atom},    
pdfauthor={AloCarvdeOliv},    
pdfnewwindow=true,      
colorlinks=true,      
linkcolor=black,         
citecolor=blue}

\usepackage{amssymb,amsfonts,amstext,amsthm}
\usepackage{mathrsfs}
\usepackage[intlimits]{amsmath}
\newtheorem{proposition}{\bf Proposition}[section]

\newtheorem{theorem}{\bf Theorem}[section]

\newtheorem{remark}{Remark}[section]

\newcommand{\LL}{{\mathrm{L}}}
\newcommand{\CC}{{\mathrm{C}}}
\newcommand{\rd}{{\mathrm{d}}}

\newcommand{\eps}{\epsilon}

\title{Hydrogen atom bound states whose spectral measures have positive upper fractal dimensions}
\author{M. Aloisio \\ {\small Departamento de Matem\'atica, UFAM, Manaus, AM, 369067-005 Brazil}\\ S. L. Carvalho \\  {\small Departamento de Matem\'atica, UFMG, Belo Horizonte, MG, 30161-970 Brazil} \\ C. R. de Oliveira\thanks{Corresponding author. Email: oliveira@dm.ufscar.br} \\ {\small Departamento de Matem\'atica, UFSCar, S\~ao Carlos, SP, 13560-970 Brazil}}
\date{March 2020}

\begin{document}

\maketitle

\begin{abstract} It is shown that, Baire generically, the bound states of the Hamiltonian of the Hydrogen atom have spectral measures with exact $0$-lower and $1/3$-upper generalized fractal dimensions; the relation to (a weak form of) dynamical delocalization along orthonormal bases is also discussed.
Such result is a consequence of the distribution of the Hamiltonian eigenvalues.
\end{abstract} 

\


\



\section{Introduction}

\subsection{Contextualization}

\noindent The Hydrogen atom (H-atom) Hamiltonian, which is given, in suitable units, by 
\[H = -\Delta - \frac{\kappa}{r},\qquad r=|x|, \;\kappa>0,\] 
and acts in the Hilbert space $\LL^2(\mathbb R^3)$, is an outstanding quantum model that combines simplicity and physical importance. Surely, $H$ and the harmonic oscillator Schr\"odinger operator are the most studied and discussed quantum models in the literature, and both models can be exactly solved. However, the H-atom Hamiltonian describes a more realistic case, and some of its (quantum) energy transitions were experimentally obtained before the birth of quantum theory and have constituted a strong support to the eigenvalue theory developed by Schr\"odinger in the 1920s. Spin is not considered here.

Notably, from the mathematical viewpoint, the general quantum framework proposed by von Neumann~\cite{vNeumann}, with the nonobvious concept of unbounded self-adjoint operators, was  supported by Kato's proof~\cite{katoHH} that~$H$, with domain equal to the space of smooth functions with compact support, has exactly one self-adjoint extension. Kato has actually also considered the Hamiltonians for other atoms and molecules.

It is well known that $H$ has a point spectrum component, with eigenvalues  of the form~\cite{williams}
\[\lambda_n =- \frac{\Lambda}{n^2},\qquad n=1,2,3,\cdots,\]
with $\Lambda=\kappa^2/4$ in the considered units; each $\lambda_n$ has multiplicity~$n^2$. Let us denote by $|n,l,m\rangle$ the corresponding orthonormalized eigenfunctions, so that, for each~$n$,
\[H|n,l,m\rangle  = \lambda_n |n,l,m\rangle ,\quad l=0,1,\cdots, n-1,\quad m=-l,\cdots,l.\] In the position representation, such eigenfunctions correspond to the famous orbitals, and any element in the closed subspace generated by them (we denote such subspace by $\dot {\mathcal H}$; see ahead) is interpreted as a {\em bound state}.

There is also a continuous spectrum filling the interval $[0,\infty)$; vectors in the continuous subspace are called {\em unbound states}, and they describe ionized H-atoms. Thus, the dynamics of vectors in such spaces are expected to be quite different. The main interest here is on the point subspace, $\dot{\mathcal H}$, generated by the eigenvectors $\{|n,l,m\rangle \}_{n,l,m}$ of~$H$, and only vectors in this (closed) subspace will be considered in what follows; the restriction of~$H$ to this subspace generated by orbitals will be denoted by~$\dot H$ (recall that $\dot{\mathcal H}$ reduces~$\dot H$; see Section~9.8 and Theorem~12.1.2 in~\cite{ISTQD}). 

The restricted (self-adjoint) operator~$\dot H$ is bounded, its spectrum is $\{\lambda_n\}_{n=1}^\infty\cup\{0\}$, and given a general initial bound state condition $|\psi\rangle=\sum_{n,l,m} a_{n,l,m}\,|n,l,m\rangle $ in~$\dot{\mathcal H}$, its action over $|\psi\rangle$ is simply
\[\dot H |\psi\rangle =- \sum_{n,l,m} \frac{\Lambda}{n^2}\,a_{n,l,m}\,|n,l,m\rangle, \]
whereas the dynamics of $|\psi\rangle$ generated by $\dot H$ is given by the solution
\[|\psi(t)\rangle=e^{-it{\dot H}}|\psi\rangle  = \sum_{n,l,m} e^{it\Lambda/n^2}\,a_{n,l,m}\,|n,l,m\rangle \]to the Schr\"odinger equation
\[i\frac{\partial v}{\partial t} = \dot Hv,\qquad v(0)=|\psi\rangle.\]
The unitarity of the time evolution implies that 
\[
\||\psi\rangle\|^2=\||\psi(t)\rangle\|^2=\langle \psi(t)|\psi(t)\rangle=\displaystyle\sum_{n,l,m}|a_{n,l,m}|^2,
\] for all~$t\in\mathbb R$.

To each state $|\phi\rangle$, one associates a (unique) positive and finite measure $\mu_{|\phi\rangle}$, the so-called {\em spectral measure of~$|\phi\rangle$} in the mathematics literature and {\em density of states} in the physics literature, characterized by
\[\langle\phi|\phi(t)\rangle = \int_{\mathbb R} e^{-itx}\,\rd \mu_{|\phi\rangle}(x),\quad \forall t\in\mathbb R.\]
Explicitly, the spectral measure of a bound state~$|\psi\rangle$ is given by a sum of Dirac deltas (which are particular cases of the so-called {\em pure point measures}), i.e.,
\[ \mu_{|\psi\rangle}(x) = \sum_{n,l,m} |a_{n,l,m}|^2 \,\delta_{\lambda_n}(x)\,,\] 
from which one may clearly see its interpretation: for a normalized $|\psi\rangle$, $\mu_{|\psi\rangle}$ is concentrated on the eigenvalues~$\lambda_n$ of~$\dot H$ whose weights are exactly the probabilities~$|a_{n,l,m}|^2$ of the state to be found in the corresponding eigenfunctions~$|n,l,m\rangle$ (thus, $\mu_{|\psi\rangle}(\{\lambda_n\})=\sum_{l,m}|a_{n,l,m}|^2$ is the probability of the state to be found in an eigenstate corresponding to the energy~$\lambda_n$).

Besides the possibility of being a combination of Dirac deltas, for general self-adjoint operators the spectral measures may be directly related to the Lebesgue measure (a particular case of the absolutely continuous spectrum, the kind of spectrum for vectors in the unbound states of the H-atom), or even to singular continuous measures (that do not occur for the H-atom), which are singular with respect to the Lebesgue measure and do not have atoms, usually related to fractal spectrum and characterized by nontrivial values of the {\em upper} and {\em lower generalized fractal dimensions} $D^\pm_\mu(q)$, $0<q\ne1$ (see ahead). Such notions of measure dimensions were introduced in the physics literature by Grassberger and Procaccia~\cite{GPcorrDim,hentschelProcaccia} to study invariant measures of strange attractors in dynamical systems; later on, they were  also applied to the study of spectral measures (see, for instance, \cite{BGT1997,BGT2001,GKT}).  

The dimensions $D_\mu^\pm(q)$ have other names, such as multifractal dimensions and Hentschel-Procaccia dimensions. A mathematical discussion of some fractal dimensions of measures may be found in Cutler's work~\cite{cutler1991} (see also the discussion in~\cite{MMR}).

The general belief is that a measure with positive dimension is concentrated (supported) on a ``fractal set''; since the spectral measures of bound states are instances of pure point measures, concentrated on the eigenvalues of~$\dot H$, and since such measures have no apparent fractal characteristics (for instance, their Hausdorff and packing dimensions are equal to zero; see~\cite{BGT2001}), it is usually expected that their dimensions vanish for all values of the parameter~$q$. However, there are exceptions for some point measures and our main goal in Section~\ref{sectDimensions}  is to point out that there is a ``big set'' of bound states of the H-atom that fits in the exceptions! Namely, in this work, we use the explicit form of the eigenvalues of the H-atom operator to show that, Baire generally, bound states have spectral measures with exact $0$-lower and $1/3$-upper generalized fractal dimension for all $0<q<1$  (see Theorem~\ref{mainTheor} for details and precise statements). Note that it is known that for pure point measures, $D_\mu^{\pm}(q) =0$ for all~$q>1$, so values in the range $0<q<1$ should be the ones considered in Section~\ref{sectDimensions}. 

\subsection{Fractal dimensions and main result}\label{sectDimensions}

\noindent To recall the concept of generalized fractal dimensions of  a Borel finite positive measure~$\mu$ on~$\mathbb R$, set $B(x;\epsilon)=(x-\eps,x+\eps)$ and consider (the integration is restricted to the sets so that $ \mu ( B(x;\epsilon))>0$)
\[I_\mu(q,\epsilon):=  \int \mu ( B(x;\epsilon))^{q-1} {\mathrm d}\mu(x),\]  
which takes into account a statistical distribution of~$\mu$; for values of~$q>1$, the main contribution to $I_\mu(q,\epsilon)$ comes from sets for which~$\mu$ charges most, whereas for~$q<1$, the main contribution is from the most rarefied sets charged by~$\mu$. 

Roughly, the generalized dimension~$D_\mu(q)$ of the measure~$\mu$ indicates how $I_\mu(q,\epsilon)$ polynomial scales for small~$\epsilon$, that is,
\[I_\mu(q,\epsilon)^{\frac{1}{q-1}}\;\sim\; \epsilon^{D_\mu(q)}.\]
The  precise definitions of the {\em lower} and {\em upper $q$-generalized fractal dimensions of~$\mu$} are, respectively,
\[D_\mu^-(q) := \liminf_{\epsilon \downarrow 0} \frac{\ln \,I_\mu(q,\epsilon)}{(q-1) \ln \epsilon}\,,\qquad D_\mu^+(q) := \limsup_{\epsilon \downarrow 0} \frac{\ln \,I_\mu(q,\epsilon)}{(q-1) \ln \epsilon}\,\]
 (note that the limit $\lim_{\epsilon \downarrow 0} \dfrac{\ln \,I_\mu(q,\epsilon)}{(q-1) \ln \epsilon}$ may not exist).

Suppose that~$\mu$ has  bounded support (that is, there exist $-\infty<a\le b<+\infty$ such that if $A$ is a Borel subset of $\mathbb{R}\setminus[a,b]$, then $\mu(A)=0$); some basic properties of such dimensions are~\cite{BGT1997}:
\begin{enumerate} 
\item The functions $ q \mapsto D_\mu^\pm(q)$ are nonincreasing and continuous on $(0,1)\cup (1,\infty)$.
\item  For all $q>0$, $q\ne1$,  one has $0 \leq D_\mu^-(q) \leq D_\mu^+(q) \leq 1$. 
\item For all $q>0$, $q\ne1$,  one has alternative limits for the dimensions:
\[D_\mu^-(q) = \liminf_{\epsilon \downarrow 0} \frac{\ln [\epsilon^{-1} \int_{\mathbb{R}} \mu (B(x,\epsilon))^q \, {\mathrm d}x]}{(q-1) \ln \epsilon} \quad {\rm  and } \quad D_\mu^+(q) = \limsup_{\epsilon \downarrow 0} \frac{\ln [\epsilon^{-1} \int_{\mathbb{R}} \mu (B(x,\epsilon))^q \, {\mathrm d}x]}{(q-1)\ln \epsilon}.\]
\end{enumerate}

Note that the spectral measure of any bound state vector $|\psi\rangle$ of the H-atom, 
\[\mu_{|\psi\rangle}(\cdot) = \sum_{n,l,m} |a_{n,l,m}|^2 \,\delta_{\lambda_n}(\cdot),\] 
has bounded support  (since it vanishes outside $[-\Lambda,0]$), so $D^\pm_{\displaystyle\mu_{|\psi\rangle}}(q)$ satisfy the properties listed above. As mentioned in the Introduction, they are examples of pure point measures, and  it is usually expected that their dimensions $ D_{\displaystyle\mu_{ |\psi\rangle}}^\pm(q)$ vanish for all values of the parameter~$q$.

We finally present the main result of this work. Recall that a {\em generic set} in a complete metric space is a countable intersection of open and dense sets; this concept constitutes the usual notion of ``typical set'' from the topological point of view and it
is intimately connected to Baire's Category Theorem.

\begin{theorem}\label{mainTheor} There exists a generic set~$\mathcal R$  of vectors  in~$\dot{\mathcal H}$ so that, for each~$|\psi\rangle\in\mathcal R$ and each $0<q<1$,
\[D_{\displaystyle\mu_{|\psi\rangle}}^-(q)=0 \qquad  and \qquad D_{\displaystyle\mu_{|\psi\rangle}}^+(q) =  \frac{1}{3}.\] 

\end{theorem}

\

\begin{remark}\label{remark1}{\rm  It is worth underlying that the positive upper fractal dimensions follow from the structure of the eigenvalues of the H-atom operator, and we have in fact a more general result. Namely, given $\alpha>0$, it follows from the same arguments presented in the proof of Theorem~\ref{mainTheor} that if~$T$ is a bounded self-adjoint operator, acting on a separable Hilbert space~$\mathcal{H}$, with pure point spectrum  so that:
\begin{enumerate}
\item its sequence of eigenvalues $\{\sigma_n\}$ converges to~$\sigma_0$ and satisfies 
\[
\displaystyle\lim_{n \to \infty} n^\alpha (\sigma_n-\sigma_0)  = L\ne0;
\]
\item $(\sigma_{n+1}-\sigma_n)$ vanishes monotonically;
\end{enumerate}
then, there exists a generic set~$\mathcal R$  of vectors  in~$\mathcal{H}$ so that, for each~$|\phi\rangle\in\mathcal R$ and each $0<q<1$,
\[D_{\displaystyle\mu_{|\phi\rangle}}^-(q)=0 \qquad \mathrm{ and} \qquad D_{\displaystyle\mu_{|\phi\rangle}}^+(q) = \frac{1}{1+ \alpha},\] 
where $\mu_{|\phi\rangle}$ denotes the spectral measure of $T$ associated with $|\phi\rangle$. We borrow  examples from the literature: given $0< \beta < 1$, let $\alpha = \frac{1 - \beta}{\beta}$. Then, by Theorem~2 in~\cite{Simonpointspc}, there exists a continuous one-dimensional Schr\"odinger operator
\[H_{V^\beta} := -\frac{\rd ^2}{\rd x^2} + V^\beta\]
acting in an appropriate domain in~$\LL^2[0,\infty)$ so that $\{1/j^\alpha\}_{j=1}^\infty$ are the eigenvalues of $H_{V^\beta}$ in $[0,\infty)$, where $V^\beta$ is a real-valued multiplication (potential) operator; denote by $\mathcal {\dot H}^\beta$ the closed subspace generated by the corresponding eigenfunctions. In this case, we  conclude that there exists a generic set~$\mathcal R^\beta$ in~$\mathcal{\dot H}^\beta$    so that, for each~$|\phi\rangle\in\mathcal R^\beta$ and each $0<q<1$,
\[D_{\displaystyle\mu_{|\phi\rangle}}^-(q)=0 \qquad and \qquad D_{\displaystyle\mu_{|\phi\rangle}}^+(q) = \beta.\]}
\end{remark} 

\begin{remark}{\rm Recently, the present authors have shown in \cite{ACdO} that for systems with thick point spectrum, Baire generally, each initial state has $0$-lower and $1$-upper generalized fractal dimensions, $0<q<1$. Here, we are in a different setting. Namely, in this paper, we show (in sense of the Remark \ref{remark1}) explicitly that, Baire generally, the exact values of the dimensions depend only on elementary properties of the eigenvalues of the H-atom operator.}
\end{remark}

\begin{remark}{\rm There are indications~\cite{alejandro} that spectral structures similar to the one presented in the H-atom occur for all atoms in the Periodic Table; if this is actually the case, then alike conclusions of Theorem~\ref{mainTheor} could be drawn for other atoms as well.}
\end{remark}

In Section~\ref{sectDynamics} some dynamical consequences of the positive fractal dimensions of some bound states will be discussed; this is a subtle subject that gives a flavor that the dynamics of such states (weakly) resemble the dynamics of  unbound states. Finally, in Section~\ref{sectionProof} the proof of Theorem~\ref{mainTheor} is presented. The ingredients of the proof are: some robust arguments developed in \cite{ACdO}, combined with a fine analysis of  generalized fractal dimensions of pure point measures of the H-atom (Proposition~\ref{mainlemma}). 


\section{Orthonormal basis moments}  \label{sectDynamics}

\noindent In order to probe dynamical (de)localization associated with an initial condition~$|\psi\rangle$ with respect to a general orthonormal basis~$\mathfrak B =\{|k\rangle \}$ of~$\dot{\mathcal H}$, one may quantify the ``travel to large dimensions~$k$'' by considering the time evolution of the {\em $p$-moments} of~$|\psi\rangle$, $p>0$, that is,
\[r_{p,\mathfrak B}^{|\psi\rangle}(t):=  \left(  \sum\nolimits_{k} |k|^p \; W_{|k\rangle ,|\psi\rangle}(t) \right)^{\frac1p}\; ,\]
where
\[W_{|k\rangle ,|\psi\rangle}(t) := \frac1t \int_0^t |\langle k\mid\psi(s)\rangle|^2\;\mathrm{d}s\]
denotes the time average probability of the particle (electron) in the state~$|\psi\rangle$ to be found in the basis vector~$|k\rangle $ at time~$t$. If one thinks of a polynomial growth $r_{p,\mathfrak B}^{|\psi\rangle}(t)\sim\; t^{\beta(p)}$, the {\em lower} and {\em upper $p$-moment growth exponents} are then naturally introduced, respectively, by
\[\beta_{|\psi\rangle}^-(p,\mathfrak B):= \liminf_{t\to\infty}\, \frac{\ln r_{p,\mathfrak B}^{|\psi\rangle}(t)}{\ln t},\qquad \beta_{|\psi\rangle}^+(p,\mathfrak B):= \limsup_{t\to\infty}\, \frac{\ln r_{p,\mathfrak B}^{|\psi\rangle}(t)}{\ln t}\,.\]
Finally, one  says that the state~$|\psi\rangle$ is {\em weakly dynamical delocalized} with respect to the basis~$\mathfrak B$ if, for some $p>0$, 
\[\beta_{|\psi\rangle}^+(p,\mathfrak B)>0\]
 (and so, $\beta_{|\psi\rangle}^+(p^\prime,\mathfrak B)>0$ for each $p^\prime>p$). The term {\em weakly} is due to the possibility of $\beta_{|\psi\rangle}^-(p,\mathfrak B)=0$, for all~$p>0$, and also because if $\beta_{|\psi\rangle}^+(p,\mathfrak B)>0$, the interpretation is that 
\[r_{p,\mathfrak B}^{|\psi\rangle}(t_j)\sim\; t_j^{\beta_{|\psi\rangle}^+(p,\mathfrak B)}\]
for a sequence of instants of time $t_j\to\infty$. Furthermore, for large~$t_j$, one concludes that the time average projections~$|\langle k|\psi(t)\rangle|^2$ are relevant for large values of~$k$, and so the dynamics resembles that of the unbounded states, since the initial state $|\psi\rangle$ ``scapes'' any finitely generated subspace of $\dot{\mathcal{H}}$ as $t_j\to\infty$ (it is a kind of dynamical instability, and for the continuous subspace of the H-atom it follows by Riemann-Lebesgue Lemma; see Theorem~13.3.7 in~\cite{ISTQD}). 

Now we may combine Theorem~\ref{mainTheor}  with an important relation  due to  Barbaroux, Germinet and Tcheremchantsev~\cite{BGT2001}, independently obtained by Guarneri and Schultz-Baldes~\cite{GuarSB1999b}, which holds for all orthonormal basis and all $p>0$, namely
\[\beta_{|\psi\rangle}^+(p,\mathfrak B) \geq   D_{\displaystyle\mu_{|\psi\rangle}}^+\biggr(\frac{1}{1+p}\biggr),\] 
in order to obtain the following result.

\begin{theorem}\label{mainTheor2}
There exists a generic set~$\mathcal R$  of initial conditions in~$\dot{\mathcal H}$ so that, for each~$|\psi\rangle\in\mathcal R$, each orthonormal basis~$\mathfrak B$ and each $p>0$, 
\[\beta_{|\psi\rangle}^+(p,\mathfrak B) \ge \frac13.\]
In particular, weak dynamical delocalization occurs  for initial conditions in a generic set, and with respect to all orthonormal basis.
\end{theorem}

\begin{remark}\end{remark}
\begin{enumerate}
\item Note that in the inequality presented in Theorem~\ref{mainTheor2}, the lower bound is independent of the basis~$\mathfrak B$ and  $\beta_{|\psi\rangle}^+(p,\mathfrak B)=\infty$ may occur; this is the case if  $r_{p,\mathfrak B}^{|\psi\rangle}(t)=\infty$ for all~$t$. It is not an easy task to control the time dependence of the moments.

\item   If the initial condition is $|\psi\rangle = |n,l,m\rangle$ (i.e., an orbital), then $|\psi(s)\rangle = e^{-is\lambda_n}|\psi\rangle$, and so $W_{|k\rangle ,|\psi\rangle}(t)$ is constant in time for all~$k$. Hence, $r_{p,\mathfrak B}^{|\psi\rangle}(t)$ is also constant, which may be finitely or infinitely valued. If it is finitely valued, then $\beta_{|\psi\rangle}^+(p,\mathfrak B)=0$ for each~$p>0$ and each~$|n,l,m\rangle\notin \mathcal R.$
\item If the basis is given by the orbitals, that is, if $\mathfrak B'=\{|n,l,m\rangle\}$, then since for each bound state $|\psi\rangle$ and each $s$, one has
\begin{eqnarray*}
|\langle n,l,m|\psi(s)\rangle| &=& |\langle n,l,m|e^{-is\dot H}\psi\rangle| \\ &=& |\langle e^{is\dot H}(n,l,m)|\psi\rangle| = |\langle e^{is\lambda_n}(n,l,m)|\psi\rangle|=|\langle n,l,m|\psi\rangle|\,,
\end{eqnarray*} it follows that  $W_{|n,l,m\rangle ,|\psi\rangle}(t)$ is also constant in time for all~$n,l,m$. Hence,~$\mathfrak B'$ is not an interesting basis here. 
\item By suitable adaptations of results in~\cite{GK2004}, given a smooth function with compact support $f\in\CC_0^\infty(\mathbb R)$, for each  $f(\dot H)|\psi\rangle$ and each~$p>0$, one has $0\leq \beta_{f(\dot H)|\psi\rangle}^+(p,\mathfrak B)\leq 1$, as soon as the normalized vector $|\psi\rangle/\||\psi\rangle\|$ is an element of~$\mathfrak B$; such bases are actually interesting in this setting.
\end{enumerate}

\begin{remark}{\rm 
With respect to the dimensions and dynamical consequences presented above,  it was not without surprise the discovery that there were basic theoretical properties of the Hydrogen atom Hamiltonian not yet explored in the literature. It is natural to wonder why  such properties were not noticed in the literature so far. We see two reasons; first,  the vectors generated by a finite set of H-atom eigenfunctions are excluded  from the generic set for which our results hold true (without mentioning the specific forms of the coefficients $a_n$; see the proofs in Section~\ref{sectionProof}); second, from the dynamical point of view, the ``instability'' is presumed to occur only for a subsequence of times. We expect that this work could stimulate experimentalists to check our findings.
 }
\end{remark}

\begin{remark}{\rm To put our work into perspective, we note that establishing bounds on the exponents  $\beta_\psi^\mp(p,\mathfrak B)$ in purely spectral terms (that is, without making reference to a specific form of a self-adjoint operator $T$) is an important problem, that has attracted significant interest (see~\cite{Alo,BGT2001,BreuerLS,Carvalho1,Carvalho2,DamanikDDP,Guarneri1,Guarneri2,GuarSB1999b,last2001,StolzIntr}). In this context, we emphasize that the lower bounds in Theorem~\ref{mainTheor2} only depend on the form of the Hydrogen atom eigenvalues (see also Remark \ref{remark1}).} 
\end{remark}


\section{Proof of Theorem~\ref{mainTheor}}\label{sectionProof}

\noindent Some preparation is required in order to present the proof of Theorem~\ref{mainTheor}. A crucial point in the proof is that besides the monotonic sequence of eigenvalues, $(\lambda_n)$, the sequence of distances between consecutive eigenvalues, $(\lambda_{n+1}-\lambda_{n})$, is also monotonic and it vanishes.

In the proof of Proposition~\ref{mainlemma}, we will properly bound $ 1/\eps \int \mu_{|\psi\rangle} (B(x,\eps))^q \, \rd x$ from above. In order to do that, we combine a particular partition of the support of $\mu_{|\psi\rangle}$ with the distribution of the eigenvalues of the H-atom operator.

\begin{remark}\label{remarkJensen}{\rm Let $(c_n)_{n \in \mathbb{N}}$ be a nonnegative sequence of real numbers such that $\sum_{n=1}^\infty c_n =  1$. Then, for each $0<r<1$ and each $N \in\mathbb{N}$,
\[\sum_{n=1}^N (c_n)^r \leq  \sum_{n=1}^N \frac{1}{n^r}.\]
Namely, since $0<r<1$, by Jensen's inequality,
\[ \sum_{n=1}^N (c_n)^r \leq  N^{1-r} \biggl(\sum_{n=1}^N c_n \biggl)^r\leq N^{1-r} \leq \sum_{n=1}^N \frac{1}{n^r}.\]
}
\end{remark}

\begin{proposition}\label{mainlemma} Let $0\ne|\psi\rangle \in \dot {\mathcal H}$. Then, 
 for each $0<q<1$, \[D_{\displaystyle\mu_{|\psi\rangle}}^+(q) \leq  \frac{1}{3}.\]
\end{proposition}

\begin{proof} We will assume, without loss of generality, that $\Vert |\psi\rangle \Vert = 1$ (so, $ \sum_{n=1}^\infty \mu_{|\psi\rangle}(\{\lambda_n\}) = 1$). Recall that $\lambda_n=-\Lambda/n^2$, $n\in\mathbb{N}$, denote the (finitely degenerated) eigenvalues of~$\dot H$. Since 
\[\displaystyle\lim_{n \to \infty} n^3 (\lambda_{n+1} - \lambda_n) = 2 \Lambda,\] 
there exists $C>0$ such that, for each $n\in\mathbb{N}$, 
\begin{equation}\label{1eqlem}
(\lambda_{n+1} - \lambda_n)  > \frac{C}{n^3}. 
\end{equation}

\ 

Fix $0<r<q<1$. For each $\epsilon>0$, let $N_{\eps}:= [(C/\epsilon)^{1/3}]$ be the integer part of $(C/\epsilon)^{1/3}$, and note that 
\begin{eqnarray*}
\nonumber \frac{1}{\eps} \int \mu_{|\psi\rangle}(B(x,\eps))^r \, \rd x &=&  \frac{1}{\eps} \int_{(-(\Lambda + \eps),-\Lambda)} \mu_{|\psi\rangle}(B(x,\eps))^r \, \rd x + \frac{1}{\eps} \int_{[-\Lambda,\lambda_{N_\eps})} \mu_{|\psi\rangle}(B(x,\eps))^r \, \rd x\\  &+&  \frac{1}{\eps} \int_{[\lambda_{N_\eps},0)} \mu_{|\psi\rangle}(B(x,\eps))^r \, \rd x  + \frac{1}{\eps} \int_{[0,\eps)} \mu_{|\psi\rangle}(B(x,\eps))^r \, \rd x.
\end{eqnarray*}
By using  that $\mu_{|\psi\rangle}(B(x,\eps))\le 1$ in the first and fourth integrals, one obtains
\begin{eqnarray}\label{2eqlem}
 \frac{1}{\eps} \int \mu_{|\psi\rangle}(B(x,\eps))^r \, \rd x \leq   \frac{1}{\eps} \int_{[-\Lambda,\lambda_{N_\eps})} \mu_{|\psi\rangle}(B(x,\eps))^r \, \rd x +  \frac{1}{\eps} \int_{[\lambda_{N_\eps},0)} \mu_{|\psi\rangle}(B(x,\eps))^r \, \rd x + 2.
\end{eqnarray}

We will separately estimate both integrals on the right side of (\ref{2eqlem}). We remark that $N_{\eps}$ was chosen so that $(\lambda_{N_{\eps}+1} - \lambda_{N_{\eps}}) > \eps$.  In what follows, $C_r$ and $C_\Lambda$ always represent constants that depend only on $r$ and $\Lambda$, respectively. We also note that the values of $C_r$ and $C_\Lambda$ may change from one estimate to another.

Since $\displaystyle\lim_{n\to\infty} (\lambda_{n+1} - \lambda_{n}) =0$ monotonically and, by (\ref{1eqlem}), $(\lambda_{N_{\eps}+1} - \lambda_{N_{\eps}}) > \eps$, it follows that for each $2\leq n\leq N_{\eps} -1$, 
\begin{eqnarray*}
\max_{x \in B(\lambda_n,\eps)} \mu_{|\psi\rangle}(B(x,\eps))^r  &\leq& \mu_{|\psi\rangle}(B(\lambda_n,2\eps))^r \leq (\mu_{|\psi\rangle}(\{\lambda_{n-1}\})  + \mu_{|\psi\rangle}(\{\lambda_n\}) + \mu_{|\psi\rangle}(\{\lambda_{n+1}\}))^r\\ &\leq& 3^{r}\max\{\mu_{|\psi\rangle}(\{\lambda_{n-1}\})^r,\mu_{|\psi\rangle}(\{\lambda_{n}\})^r,\mu_{|\psi\rangle}(\{\lambda_{n+1}\})^r\}\\&\leq &  3^r\big[ \mu_{|\psi\rangle}(\{\lambda_{n-1}\})^r + \mu_{|\psi\rangle}(\{\lambda_{n}\})^r + \mu_{|\psi\rangle}(\{\lambda_{n+1}\})^r\big].
\end{eqnarray*}
Thus, by Remark~\ref{remarkJensen}, 

\begin{eqnarray}\label{3eqlem}
\nonumber  \frac{1}{\eps} \sum_{n=1}^{{N_{\eps}}-1} \int_{B(\lambda_n,\eps)} \mu_{|\psi\rangle}(B(x,\eps))^r \, \, \rd x &\leq& 2\big [\mu_{|\psi\rangle}(\{\Lambda\}) + \mu_{|\psi\rangle}(\{\lambda_2\})\big] +   C_r  \sum_{n=1}^{{N_{\eps}}} \mu_{|\psi\rangle}(\{\lambda_n\})^r\\ &\leq& C_\Lambda + C_r \sum_{n=1}^{{N_{\eps}}} \frac{1}{n^r}. 
\end{eqnarray}
Now, for each $x \in [-\Lambda,\lambda_{N_{\eps}})$, 
\begin{eqnarray}\label{4eqlem} \nonumber x \not \in \displaystyle\bigcup_{n=1}^{N_{\eps}} B(\lambda_n,\eps) & \Longleftrightarrow& |x- \lambda_n| \geq \eps, ~{\rm  for}~  1\leq n\leq N_{\eps} \\ &\Longleftrightarrow&  \nonumber \lambda_n  \not \in B(x,\eps), ~{\rm  for}~ 1\leq n \leq N_{\eps} \\ &\Longleftrightarrow&   \mu_{|\psi\rangle}(B(x,\eps)) = 0.
\end{eqnarray}

Thus, by combining 
(\ref{3eqlem}) and (\ref{4eqlem}), one has, 
\begin{eqnarray}\label{5eqlem}
\nonumber   \frac{1}{\eps} \int_{[-\Lambda,\lambda_{N_{\eps}})} \mu_{|\psi\rangle}(B(x,\eps))^r \, \, \rd x &\leq&  \frac{1}{\eps} \sum_{n=1}^{{N_{\eps}}-1} \int_{B(\lambda_n,\eps)} \mu_{|\psi\rangle}(B(x,\eps))^r \, \, \rd x +  \frac{1}{\eps} \int_{[\lambda_{N_{\eps} - 1},\lambda_{N_{\eps}})} \mu_{|\psi\rangle}(B(x,\eps))^r \, \, \rd x\\ \nonumber &\leq& C_\Lambda +  C_r \sum_{n=1}^{{N_{\eps}}} \frac{1}{n^r} + \frac{1}{\eps}\vert \lambda_{N_{\eps}} - \lambda_{N_{\eps} - 1} \vert   \\ \nonumber &\leq& C_\Lambda + C_r\int_1^{N_\eps} \frac{1}{x^r} \rd x  \\ \nonumber &\leq& (C_\Lambda + C_r) \, \frac{N_\eps^{1-r} - 1}{1-r} \\ &\leq& (C_\Lambda + C_r)   \,  \frac{\eps^{\frac{r-1}{3}} - 1}{1-r}.
\end{eqnarray}

Since $\displaystyle\lim_{n \to \infty} n^3 (\lambda_{n+1} - \lambda_n) = 2 \Lambda$, it follows that, for sufficiently small $\eps$,
\begin{eqnarray}\label{6eqlem}\nonumber  \frac{1}{\epsilon} \int_{[\lambda_{N_{\eps}},0)} \mu_{|\psi \rangle}(B(x,\eps))^r\, \rd x &=&  \frac{1}{\epsilon} \sum_{n=N_{\eps}}^\infty \,\int_{[\lambda_n, \lambda_{n+1})}  \mu_{|\psi \rangle}(B(x,\eps))^{r} \, \rd x \\ \nonumber &\leq&  \frac{1}{\epsilon} \sum_{n=N_{\eps}}^\infty |\lambda_{n+1} -\lambda_n|\\\nonumber &\leq& C_\Lambda  \frac{1}{\epsilon} \sum_{n=N_{\eps}}^\infty \frac{1}{n^{3}}\\ \nonumber &\leq& C_\Lambda     \frac{1}{\epsilon} \int_{N_{\eps} - 1}^\infty  \frac{1}{x^{3}} \, \rd x\\ \nonumber &\leq&  C_\Lambda     \frac{1}{\epsilon} \frac{1}{(N_{\eps} -1)^{2}} \\ \nonumber &\leq& C_\Lambda  \, \eps^{-\frac{1}{3}} \\ &=& C_\Lambda   \, \eps^{\frac{r-1}{3(1-r)}}.
 \end{eqnarray} 

By combining (\ref{2eqlem}), (\ref{5eqlem}) and (\ref{6eqlem}), one obtains, for sufficiently small $\eps$,  

\begin{equation*}L_{\displaystyle\mu_{|\psi \rangle}} (r,\epsilon) := \frac{1}{\epsilon} \int \mu_{|\psi \rangle}(B(x,\eps))^{r} \, \rd x \leq  (C_\Lambda + C_r)  \, \eps^{\frac{r-1}{3(1-r)}}=: Z(\Lambda,r,\eps). 
\end{equation*}
 Hence, 
\[D_{{\mu_{|\psi \rangle}}}^+(r) = \limsup_{\epsilon \downarrow 0} \frac{\ln \,L_{\displaystyle\mu_{|\psi \rangle}} (r,\epsilon)}{(r-1) \ln \epsilon} \leq \limsup_{\epsilon \downarrow 0} \frac{\ln \, Z(\Lambda,r,\eps)}{(r-1) \ln \epsilon} = \frac{1}{3(1-r)}.\]

Finally, since $r\mapsto D_\mu^+(r)$  is a nonincreasing function (see Subsection~\ref{sectDimensions}), one has

\[D_{{\displaystyle\mu_{|\psi \rangle}}}^+(q) \leq \limsup_ {r \to 0} D_{{\mu_{|\psi \rangle}}}^+(r)  \leq \limsup_ {r \to 0}\frac{1}{3(1-r)}= \frac{1}{3}\,,
\] for all $0<q<1$.
\end{proof}

\begin{proposition}[Proposition 3.1 in \cite{ACdO}]\label{proposition2} For every $q\in(0,1)$ and every $0\le\gamma\le1$, each of the sets
\[G_{\gamma}^{-}(q):=\big\{|\psi\rangle \in \dot{\mathcal{H}} \mid D_{\displaystyle\mu_{|\psi\rangle}}^{-}(q) \leq \gamma\big\}\]
and
\[G_{\gamma}^{+}(q):=\big\{|\psi\rangle \in \dot{\mathcal{H}} \mid D_{\displaystyle\mu_{|\psi\rangle}}^{+}(q) \geq \gamma\big\} \]
is a $G_\delta$ set in~$\dot{\mathcal H}$.
\end{proposition}

\begin{proof}[{Proof} {\rm (Theorem~\ref{mainTheor})}]
Note that in the proof  of the lower dimension part,  we only use the properties of separability of~$\dot{\mathcal H}$ and the fact that the spectrum of~$\dot H$ is bounded and  pure point.
 
 Fix $0<q<1$ and pick a sequence of complex numbers $b_n\ne0$ such that 
\[ \sum_n |b_n|^{2q}<\infty;\]
hence, $\sum_n |b_n|^{2}<\infty$. Consider an arbitrary bound state  $|\psi\rangle = \sum_{n,l,m} a_{n,l,m}|n,l,m\rangle\in \dot{\mathcal H}$, together with the following sequences of vectors, for $k\in\mathbb N$:
\[|\phi_k\rangle = \sum_{n=1}^{k-1} \sum_{l,m} a_{n,l,m}|n,l,m\rangle,\quad |\psi_k\rangle = \sum_{n=1}^{k-1} \sum_{l,m} a_{n,l,m}|n,l,m\rangle + \sum_{n=k}^\infty b_n\,|n,0,0\rangle.\]
The hypothesis on~$b_n$ and the identity $\||\psi\rangle\|^2=\sum_{n,l,m}|a_{n,l,m}|^2$ result in 
\[\lim_{k\to\infty} |\phi_k\rangle = |\psi\rangle, \qquad \lim_{k\to\infty} (|\phi_k\rangle-|\psi_k\rangle)=0,\]
and so $\lim_{k\to\infty} |\psi_k\rangle = |\psi\rangle$; hence, by Proposition \ref{proposition2}, it is enough to show that $D_{\displaystyle\mu_{|\psi_k\rangle}}^-(q)=0$, for all~$|\psi_k\rangle$, to conclude that $G_{0}^{-}(q)$ is generic.

Now, two simple properties (and minor variations)  will be used: $\mu_{|\psi\rangle}(\{\lambda_n\})=\sum_{l,m}|a_{n,l,m}|^2$ and, since $0<q<1$, one has for all~$\eps>0$, $\mu_{|\psi\rangle}(B(x,\eps))^{q-1}\le \mu_{|\psi\rangle}(\{\lambda_n\})^{q-1}$. Hence,
\begin{eqnarray*}
I_{\displaystyle\mu_{|\psi_k\rangle}}(q,\eps) &=& \int \mu_{|\psi_k\rangle}(B(x,\eps))^{q-1} \rd \mu_{|\psi_k\rangle}(x) \\ &=& \sum_{n\ge 1}\mu_{|\psi_k\rangle}(B(\lambda_n,\eps))^{q-1}\,\mu_{|\psi_k\rangle}(\{\lambda_n\}) \\ &\le& \sum_{n\ge 1} \mu_{|\psi_k\rangle}(\{\lambda_n\})^q \\ &=& \sum_{n=1}^{k-1} \mu_{|\psi_k\rangle}(\{\lambda_n\})^q + \sum_{n=k}^{\infty} \mu_{|\psi_k\rangle}(\{\lambda_n\})^q \\ &=& \sum_{n=1}^{k-1} \big( \sum_{l,m}|a_{n,l,m}|^2 \big)^q + \sum_{n=k}^{\infty} |b_n|^{2q} =: S(q,\psi_k)<\infty,
\end{eqnarray*} which implies
\[D_{\displaystyle\mu_{|\psi_k\rangle}}^-(q) \le \liminf_{\eps\downarrow 0} \frac{\ln S(q,\psi_k)}{(1-q)\ln\eps}=0,\] 
and so, $G_0^-(q)$ is generic. Since a countable intersection of generic sets is also generic, in order to complete the proof of the first part of the theorem, it is enough to pick $\mathcal R_- = \cap_j G_0^-(1/j)$ as the generic set of bound states;  we have used the fact that $q\mapsto D_\mu^-(q)$ is a nonincreasing function. 

Now, we proceed to the proof of the upper dimension part, for which the only used ingredients are:  $\dot{\mathcal H}$ is separable and~$\dot H$ is a bounded pure point operator with finitely degenerated eigenvalues $\lambda_n=-\Lambda/n^2$, $n\in\mathbb{N}$. Specifically, we will use the following:  i)~$\displaystyle\lim_{n\to\infty} (\lambda_{n+1} - \lambda_{n}) =0$ monotonically and,  ii)~since  $\displaystyle\lim_{n\to\infty} n^3(\lambda_{n+1} - \lambda_{n}) = 2\Lambda$, there exists a $K > 0$ such that, for large~$n$, 
\[\frac{K}{2n^3} \leq \lambda_{n+1} - \lambda_{n} \leq \frac{K}{n^3}.\]

 Fix  $0<q<1$ and let $j\in\mathbb{N}$ be such that $j > \frac{q}{1-q}$. Let also
\[|\psi_j\rangle:=\sum_n \frac{1}{\sqrt{n^{ 1 + \frac{1}{j}}}}\,|n,0,0\rangle,\]
so that $|\psi_j\rangle\in\dot{\mathcal H}$, and note that its spectral measure is 
\[\mu_{|\psi_j\rangle} = \sum_n \frac{1}{{n^{ 1 + \frac{1}{j}}}}\, \delta_{\lambda_n}.\]
For large $N$, choose    
\[\eps_N:=\frac12 (\lambda_{N+1}-\lambda_{N})\ge \frac{K}{4N^{3}}\;\Longrightarrow\; N\ge \left( \frac{K}{4\eps_N} \right)^{\frac1{3}};\] 
the monotonicity of the sequence $(\lambda_{n+1}-\lambda_{n})$ implies that $\mu_{|\psi_j\rangle}(\lambda_n-\eps_N,\lambda_n+\eps_N)=\mu_{|\psi_j\rangle}(\{\lambda_n\})$, for $1\le n\le N$. Thus, one has
\begin{eqnarray*}
I_{\displaystyle\mu_{|\psi_j\rangle}}(q,\eps_N) &=& \int \mu_{|\psi_j\rangle}(x-\eps_N,x+\eps_N)^{q-1}\;\rd \mu_{|\psi_j\rangle}(x) \ge \sum_{n=1}^N \mu_{|\psi_j\rangle}(\{\lambda_n\})^q \\ &=&  \sum_{n=1}^N \frac{1}{{n^{ (1 + \frac{1}{j})q}}} \ge \int_1^N \rd x\,\frac1{x^{(1 + \frac{1}{j})q}} = \frac{1}{1-(1 + \frac{1}{j})q}\left( N^{1-(1 + \frac{1}{j})q}- 1 \right) \\ &\ge& \frac{1}{1-(1 + \frac{1}{j})q}\left( \left( \frac{K}{4\eps_N} \right)^{\frac{1-(1 + \frac{1}{j})q}{3}}- 1 \right)\,.
\end{eqnarray*}
Since 
\[D_{\displaystyle\mu_{|\psi_j\rangle}}^+(q) \ge \limsup_{\eps_N\downarrow 0} \frac{\ln I_{\displaystyle\mu_{|\psi_j\rangle}}(q,\eps_N)}{(q-1)\ln \eps_N}\,,\] the inequality
\[D_{\displaystyle\mu_{|\psi_j\rangle}}^+(q) \geq  \frac{1-(1 + \frac{1}{j})q}{3(1-q)}\]
follows.

Let $|\psi\rangle = \sum_{n,l,m} a_{n,l,m}|n,l,m\rangle$ be a general bound state, and for $\sigma>0$, let
\[|\psi_\sigma^j\rangle = \sum_{n\le M_1} \sum_{l,m} a_{n,l,m}\,|n,l,m\rangle + \sum_{n\ge M_2} \frac{1}{\sqrt{n^{ 1 + \frac{1}{j}}}}\,|n,0,0\rangle,\]with $M_1(j) < M_2(j)$ large enough so that
\[\sum_{n> M_1}|a_{n,l,m}|^2 <\sigma^2,\qquad \sum_{n\ge M_2} \frac{1}{{n^{ 1 + \frac{1}{j}}}}<\sigma^2.\] Then, \[
\mu_{|\psi_\sigma^j\rangle} = \sum_{n\le M_1} \sum_{l,m} |a_{n,l,m}|^2\,\delta_{\lambda_n} + \sum_{n\ge M_2} \frac{1}{{n^{ 1 + \frac{1}{j}}}}\,\delta_{\lambda_n},\] and
\begin{itemize}
\item $\||\psi\rangle - |\psi_\sigma^j\rangle\|^2 < 2\sigma^2,$ 
\item $D_{\displaystyle\mu_{|\psi_\sigma^j\rangle}}^+(q) \ge  \frac{1-(1 + \frac{1}{j})q}{3(1-q)} =: \gamma_q^j$.
\end{itemize} The first item  is a simple verification, and the second one may be reduced to be above calculation for~$\mu_{|\psi_j\rangle}$ as follows. Let $\eps_N$ be as above, with $N>M_2$; then,
\[I_{\displaystyle\mu_{|\psi_\sigma^j\rangle}}(q,\sigma_N) \ge  \sum_{n= M_2}^N \frac{1}{{n^{( 1 + \frac{1}{j})q}}},\]
and the second item follows from the same estimate presented for $I_{\displaystyle\mu_{|\psi_j\rangle}}(q,\sigma_N)$.  Since this holds for all $\sigma>0$, it follows that $G_{\gamma_q^j}^{+}(q)$, with 
\[
\gamma_q^j= \frac{1-(1 + \frac{1}{j})q}{3(1-q)},
\] is a dense subset of $\dot{\mathcal{H}}$; in particular, by Proposition \ref{proposition2}, 
\[
G_{\frac{1}{3}}^{+}(q) = \bigcap_{j > \frac{q}{1-q}} G_{\gamma_q^j}^+(q)
\] is generic. Now, just pick $\mathcal R_+=\cap_l G_{\frac{1}{3}}^+(1 - 1/l)$ and use the fact that $q\mapsto D_\mu^+(q)$ is a nonincreasing function. Thus, the set
\[
\mathcal R:=\mathcal R_-\cap \mathcal R_+ = \left\{|\psi\rangle\in\dot{\mathcal H}\mid D_{\displaystyle\mu_{|\psi\rangle}}^{-}(q) =0 \quad\mathrm{and}\quad D_{\displaystyle\mu_{|\psi\rangle}}^{+}(q) \ge1/3, \quad \forall\, 0<q<1  \right\}
\]  is  also generic. Since, by Proposition~\ref{mainlemma}, $D_{\displaystyle\mu_{|\psi\rangle}}^{+}(q) \le 1/3$ holds true for all nonzero vectors in~$\dot{\mathcal H}$ and all $0<q<1$, the result follows. 
\end{proof}


\subsubsection*{Acknowledgments} 
CRdO thanks the partial support by CNPq (a Brazilian government agency, under contract 303503/2018-1). SLC thanks the partial support by Fapemig (Minas Gerais state  agency; Universal Project 001/17/CEX-APQ-00352-17).



%
%

%
%
%


\begin{thebibliography}{99}

\bibitem{Alo} M. Aloisio, A note on spectrum and quantum dynamics. J. Math. Anal. Appl. {\bf 478} (2019), 595--603.

\bibitem{ACdO} M. Aloisio, S. L. Carvalho and C. R. de Oliveira, Quantum quasiballistic dynamics and thick point spectrum. Lett. Math. Phys.  {\bf109} (2019), 1891--1906.

\bibitem{BGT1997} J.-M. Barbaroux, F. Germinet and  S. Tcheremchantsev, Generalized fractal dimensions: equivalence and basic properties. J. Math. Pure Appl. {\bf 80} (1997), 977--1012. 

\bibitem{BGT2001} J.-M. Barbaroux, F. Germinet and  S. Tcheremchantsev, Fractal dimensions and the phenomenon of intermittency in quantum dynamics. Duke Math. J. {\bf 110} (2001), 161--194. 

\bibitem{BreuerLS} J. Breuer, Y. Last, and Y. Strauss, Eigenvalue spacings and dynamical upper bounds for discrete one-dimensional Schr\"odinger operators. Duke Math. J. {\bf 157} (2011), 425--460.

\bibitem{Carvalho1} S. L. Carvalho and C. R. de Oliveira, Correlation Wonderland theorems. J.  Math. Phys. {\bf 57}  (2016), 063501.

\bibitem{Carvalho2} S. L. Carvalho and C. R. de Oliveira, Generic quasilocalized and quasiballistic discrete Schr\"odinger operators. Proc. Amer. Math. Soc. {\bf 144} (2015), 129--141.

\bibitem{cutler1991} C. D. Cutler, Some results on the behaviour and estimation of the fractal dimensions of distributions on attractors.  J. Stat. Phys.   {\bf62} (1991), 651--708.                  

\bibitem{DamanikDDP}  D. Damanik,  Schr\"odinger operators with dynamically defined potentials. Ergod. Th. \& Dynam. Syst.  {\bf 37}, 1681--1764 (2017).

\bibitem{ISTQD} C. R. de Oliveira, Intermediate Spectral Theory and Quantum Dynamics.  Birkh\"auser, PMP 54, Basel, 2009. 

\bibitem{alejandro} W. N. Favaro and A. L\'opez-Castillo, Probing the (empirical) quantum structure embedded in the Periodic Table with an effective Bohr model. Quim. Nova  {\bf36} (2013), 335--339.

\bibitem{GKT} F. Germinet, A. Kiselev and S. Tcheremchantsev: Transfer matrices and transport for Schr\"odinger operators. Ann. Inst. Fourier {\bf54} (2004),  787--830 
  
\bibitem{GK2004} F. Germinet and A. Klein, A characterization of the Anderson metal-insulator transport transition. Duke Math. J. {\bf124} (2004), 309--350. 


\bibitem{GPcorrDim} P. Grassberger and I. Procaccia, Characterization of strange attractors. Phys. Rev. Lett. {\bf50} (1983), 346--349. 

\bibitem{Guarneri1} I. Guarneri, Spectral properties of quantum diffusion on discrete lattices. Europhys. Lett. {\bf 10} (1989), 95--100.

\bibitem{Guarneri2} I. Guarneri and H. Schulz-Baldes, Lower bounds on wave-packet propagation by packing dimensions of spectral measures. Math. Phys. Elect. J. {\bf 5} (1999), 1--16.

\bibitem{GuarSB1999b}  I. Guarneri, H. Schulz-Baldes,  Intermittent lower bound on quantum diffusion. Lett. Math. Phys. {\bf49} (1999), 317--324.   

\bibitem{hentschelProcaccia} H. G. E. Hentschel and I. Procaccia, The infinite number of generalized dimensions of fractals and strange attractors. Physica {\bf 8} (1983), 435--444. 

\bibitem{last2001} Y. Last, Quantum dynamics and decomposition of singular continuous spectra. J. Funct. Anal. {\bf 142} (2001), 406--445.



\bibitem{katoHH} T. Kato, Fundamental properties of Hamiltonian operators of Schr\"odinger type. Trans. Amer. Math. Soc. {\bf70} (1951), 195--211. 

\bibitem{MMR} P. Mattila, M. Mor\'an and J.-M. Rey, Dimension of a measure. Studia Math. {\bf 142} (2000), 219--233.

\bibitem{Simonpointspc} B. Simon, Some Schr\"odinger operators with dense point spectrum. Proc. Amer. Math. Soc. {\bf125}, (1997) 203--208.

\bibitem{StolzIntr} G. Stolz, An introduction to the mathematics of Anderson Localization. Entropy and the Quantum II. Contemp. Math. {\bf 552} (2011), 71--108.

\bibitem{vNeumann}  J. von Neumann,  Mathematische Grundlagen der Quantenmechanik.  Springer-Verlag, Berlin, 1932.

\bibitem{williams} F. Williams,  Topics in Quantum Mechanics. Birkh\"auser, PMP 27, Boston, 2003.

\end{thebibliography}
\end{document}